\documentclass[submission,copyright,creativecommons]{eptcs}

\usepackage[T1]{fontenc}
\usepackage{amsmath, amssymb, amsthm, color, xcolor, pifont, array, colortbl}
\usepackage{graphicx}
\usepackage{float}
\usepackage{fancyvrb}
\usepackage{csquotes}

\newcommand{\betweenSymbol}{\raisebox{0.4ex}{\ensuremath{\mathord{\includegraphics[width=0.7em]{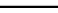}}}}}

\newcommand{\bT}[3]{{#1} \betweenSymbol {#2} \betweenSymbol {#3}}
\newcommand{\congT}[2]{{#1} \mathbin{\equiv} {#2}}
\DeclareMathOperator{\coll}{Col}
\newcommand{\col}[3]{\coll \, {#1} \, {#2} \, {#3}}
\newcommand{\para}[2]{{#1}\mathbin{\parallel} {#2}}

\newcommand{\imp}{\Rightarrow}

\newtheorem{theorem}{Theorem}

\DefineVerbatimEnvironment
{CoqVerbatim}{Verbatim}
{frame=lines, fontsize=\small}

\title{Towards an Independent Version of\\ Tarski's System of Geometry}

\author{Pierre Boutry
\institute{Centre Inria d’Universit\'e C\^ote d’Azur, Sophia Antipolis, France}
\email{pierre.boutry@inria.fr}
\and
St\'ephane Kastenbaum
\institute{No affiliation}
\email{stephane.kastenbaum@gmail.com}
\and
Cl\'ement Saintier
\institute{No affiliation}
\email{clement.saintier@gmail.com}}

\newcommand{\titlerunning}{Towards an Independent Version of Tarski's System of Geometry}
\newcommand{\authorrunning}{P. Boutry et al.}

\hypersetup{
  bookmarksnumbered,
  pdfauthor   = {\authorrunning},
  pdfsubject  = {EPTCS},               % Consider adding a more appropriate subject or description
  pdfkeywords = {Independence, Coq, Formalization of geometry} % Uncomment and enter keywords specific to your paper
}

\begin{document}

\maketitle

\setcounter{footnote}{0}% Reset footnote counter

\begin{abstract}
In 1926--1927, Tarski designed a set of axioms for Euclidean geometry which reached its final form in a manuscript by Schwabh\"auser, Szmielew and Tarski in 1983.
The differences amount to simplifications obtained by Tarski and Gupta.
Gupta presented an independent version of Tarski's system of geometry, thus establishing that his version could not be further simplified without modifying the axioms.
To obtain the independence of one of his axioms, namely Pasch's axiom, he proved the independence of one of its consequences: the previously eliminated symmetry of betweenness.
However, an independence model for the non-degenerate part of Pasch's axiom was provided by Szczerba for another version of Tarski's system of geometry in which the symmetry of betweenness holds.
This independence proof cannot be directly used for Gupta's version as the statements of the parallel postulate differ.

In this paper, we present our progress towards obtaining an independent version of a variant of Gupta's system.
Compared to Gupta's version, we split Pasch's axiom into this previously eliminated axiom and its non-degenerate part and change the statement of the parallel postulate.
We verified the independence properties by mechanizing counter-models using the Coq proof-assistant.
\end{abstract}

\section{Introduction}

The independence\footnote{We recall that an axiom is said to be independent from a set of axiom if it is not derivable from the axioms in this set.} of axioms for geometry has often been an important topic in the field of geometry.
For centuries, many mathematicians believed that Euclid’s fifth postulate was rather a theorem which could be derived from the first four of Euclid’s postulates.
History is rich with incorrect proofs of Euclid’s fifth postulate.
In 1763, Kl\"ugel provided, in his dissertation, a survey of about 30 attempts to \enquote{prove Euclid’s parallel postulate}~\cite{klugel63}.
The question was finally settled in 1832 and 1840, when Bolyai~\cite{bolyai1832appendix} and Lobachevsky~\cite{lobachevsky1840geometrische} exhibited models of hyperbolic geometry, thus establishing that this postulate was independent.
Later, Hilbert dedicated the second section of his famous \emph{Grundlagen der Geometrie}~\cite{hilbert_les_1971} to independence properties.
Then, when working out the final version~\cite{givant99} of the axioms for \emph{Metamathematische Methoden in der Geometrie}~\cite{tarski83}, commonly referred to as SST, independence results proved very helpful.

So it is a surprise that, now that we have access to tools like proof assistants which we believe to be perfectly suited for the task, the only independence to be mechanized was the one for Euclid’s fifth postulate~\cite{narboux:hal-01779452}.
To the best of our knowledge, the most recent work on the topic is the formalization of the Poincar\'e disk model in Isabelle/HOL~\cite{simic-poincare}.

In this paper, we study independence properties linked to SST~\cite{tarski83}.
SST has the advantage of being expressed in the first-order language rather than natural language which leaves room for interpretation leading to possible problems~\cite{braun:hal-01332044}.\footnote{A possible interpretation of Hilbert's axiom could lead to a degenerate model for first two groups of Hilbert’s axioms.}
There are several ways to prove independence results~\cite{beeson_herbrands_2015}.
Here we focus on independence through counter-model, i.e. constructing a model where the axiom to be proven independent will not hold while all the others will.

In 1965, Gupta presented an independent version of Tarski's system of geometry~\cite{gupta_contributions_1965}.
To obtain the independence of one of his axioms, namely Pasch's axiom, he proved the independence of one of its consequences: the previously eliminated symmetry of betweenness.
However, an independence model for the non-degenerate part of Pasch's axiom was provided by Szczerba for another version of Tarski's system of geometry in which the symmetry of betweenness holds~\cite{szczerba_independence_1970}.
This independence proof cannot be directly used for Gupta's version as the statements of the parallel postulate differ.
This can be remedied by carefully choosing the statement of the parallel postulate amongst the ones known to be equivalent~\cite{Boutry2017}.
We aim to verify that splitting Pasch's axiom into its non-degenerate part and the symmetry of betweenness in addition to changing the statement of the parallel postulate allows to obtain a system that is still independent.
So we go in the opposite direction of what Makarios did by removing the need of the reflexivity properties for congruence thanks to a modification to the five-segment property~\cite{makarios_further_2013}.
Indeed, our view is that an axiom should capture a limited and well-defined property, instead of trying to minimize the needed number of axioms at any cost.

We remark that a small change in the statement of an axiom can change whether or not it holds in a specific model.
This makes a computer very well suited to the verification that an axiom holds in a model.
So we chose to mechanize the various counter-models needed for this task in the Coq proof-assistant~\cite{the_coq_development_team_2018_1174360}.

The rest of the paper is structured as follows.
First, in Sec.~\ref{tarski}, we present the system we will be working on throughout the rest of this paper.
Then, in Sec.~\ref{model}, we show how to build a model of Tarski's axiom.
Finally, before concluding on the achieved results, we present an example of independence proof in Sec.~\ref{klein}.

\section{A Variant of Tarski's System of Geometry}\label{tarski}

In this section, we start by recalling the axioms of Tarski's system of geometry.
Then, we present the modification Gupta made to it to obtain a fully independent system~\cite{gupta_contributions_1965}.
Finally, we describe how to modify his system to combine his results and the ones from Szczerba~\cite{szczerba_independence_1970}.

\subsection{Tarski's System of Geometry}

Tarski's axiom system is based on a single primitive type depicting points and two predicates, namely congruence and betweenness.
$\congT{AB}{CD}$ states that the segments $\overline{AB}$ and $\overline{CD}$ have the same length.
$\bT{A}{B}{C}$ means that $A$, $B$ and $C$ are collinear and $B$ is between $A$ and $C$ (and $B$ may be equal to $A$ or $C$).
For an explanation of the axioms and their history see~\cite{givant99}.
Table~\ref{tarski-axiomatique-formalisee} lists the axioms for planar Euclidean geometry.

\begin{table}
\begin{center}
\begin{tabular}{crl}
A1 & Symmetry & $\congT{AB}{BA}$\\
A2 & Pseudo-Transitivity & $\congT{AB}{CD} \land \congT{AB}{EF} \imp \congT{CD}{EF}$\\
A3 & Cong Identity & $\congT{AB}{CC} \imp A=B$\\
A4 & Segment construction &$\exists E, \bT{A}{B}{E} \land \congT{BE}{CD}$\\
A5 & Five-segment & $\congT{AB}{A'B'} \land \congT{BC}{B'C'} \land$\\
   &               & $\congT{AD}{A'D'} \land \congT{BD}{B'D'} \land$\\
   &               & $\bT{A}{B}{C} \land \bT{A'}{B'}{C'} \land A \neq B \imp \congT{CD}{C'D'}$\\
A6 & Between Identity & $\bT{A}{B}{A} \imp A=B$\\
A7 & Inner Pasch & $\bT{A}{P}{C} \land \bT{B}{Q}{C} \imp \exists X,\bT{P}{X}{B} \land \bT{Q}{X}{A}$\\
A8 & Lower Dimension &$\exists A B C, \lnot \bT{A}{B}{C} \land \lnot \bT{B}{C}{A} \land \lnot \bT{C}{A}{B}$\\
A9 & Upper Dimension &$\congT{AP}{AQ} \land \congT{BP}{BQ} \land \congT{CP}{CQ} \land P \neq Q \imp$\\
&&$\bT{A}{B}{C} \lor \bT{B}{C}{A} \lor \bT{C}{A}{B}$\\
A10 & Euclid & $\bT{A}{D}{T} \land \bT{B}{D}{C} \land A \neq D \imp$\\
&&$\exists X Y, \bT{A}{B}{X} \land \bT{A}{C}{Y} \land \bT{X}{T}{Y}$\\
A11 & Continuity & $(\exists A, (\forall X Y, \Xi (X) \land \Upsilon (Y) \imp \bT{A}{X}{Y})) \imp$\\
&&$~\exists B, (\forall X Y, \Xi (X) \land \Upsilon (Y) \imp \bT{X}{B}{Y})$
\end{tabular}
\caption{Tarski's axiom system for planar Euclidean geometry.}
\label{tarski-axiomatique-formalisee}
\end{center}
\end{table}

\subsection{Gupta's Contribution}

The problem of the independence of Tarski's axiom system, as defined in Table~\ref{tarski-axiomatique-formalisee}, remains open.
Let us introduce the modifications Gupta made to it to obtain an independent system.
He reintroduced the inner transitivity of betweenness A15\footnote{We number them as in~\cite{givant99}.} in Table~\ref{tarski-variant}.
Having added this axiom, the identity axiom for betweenness A6 became a theorem and could then be removed from the system.
Finally A2, A9 and A11 are replaced by A2', A9' and A11'.
We omit the details of how to mechanize in Coq that this system, consisting of A1, A2', A3-A5, A7, A8, A9', A10, A11' and A15,\footnote{Actually the statement for A7 differs in~\cite{givant99} but the change is not important here.} and Tarski's system are equivalent.
Gupta proves that this system is independent.
To prove that A7 is independent in this system, he shows that A14, a consequence of A7 in this system, does not hold.
However, Szczerba~\cite{szczerba_independence_1970} found that A7 does not hold when A14 and all the other axioms, with the exception of A10, in Gupta's system do.
So this would suggest that A7 can be split into A14 and a variant of A7 while still having an independent system.

\begin{table}
\begin{center}
\begin{tabular}{crl}
A0 & Point equality decidability & $X = Y \lor X \not = Y$\\
A2' & Pseudo-Transitivity & $\congT{AB}{EF} \land \congT{CD}{EF} \imp \congT{AB}{CD}$\\
A7' & Inner Pasch & $\bT{A}{P}{C} \land \bT{B}{Q}{C} \land$\\
&&$A \neq P \land  P \neq C \land B \neq Q \land Q \neq C \land$\\
&&$\lnot \left( \bT{A}{B}{C} \lor \bT{B}{C}{A} \lor \bT{C}{A}{B} \right) \imp$\\
&&$\exists X,\bT{P}{X}{B} \land \bT{Q}{X}{A}$\\
A9' & Upper Dimension &$\congT{AP}{AQ} \land \congT{BP}{BQ} \land \congT{CP}{CQ} \land$\\
&&$P \neq Q \land A \neq B \land A \neq C \land B \neq C \imp$\\
&&$\bT{A}{B}{C} \lor \bT{B}{C}{A} \lor \bT{C}{A}{B}$\\
A10' & Proclus & $\para{AB}{CD} \land \col{A}{B}{P} \land \lnot \col{A}{B}{Q} \imp$\\
&&$\exists Y, \col{C}{D}{Y} \land \col{P}{Q}{Y}$\\
A11' & Continuity & $(\exists A, (\forall X Y, \Xi (X) \land \Upsilon (Y) \imp \bT{A}{X}{Y})) \imp$\\
&&$~\exists B, (\forall X Y, \Xi (X) \land \Upsilon (Y) \imp$\\
&&$~~~~~~~~X = B \lor B = Y \lor \bT{X}{B}{Y})$\\
A14 & Between Symmetry & $\bT{A}{B}{C} \imp \bT{C}{B}{A}$\\
A15 & Between Inner Transitivity & $\bT{A}{B}{D} \land \bT{B}{C}{D} \imp \bT{A}{B}{C}$
\end{tabular}
\caption{Added axioms to Tarski's system of geometry.}
\label{tarski-variant}
\end{center}
\end{table}

\subsection{An Independent System for Planar Geometry}\label{independent-tarski}

The system that we want to prove independent is very close to the one Gupta studied in his thesis~\cite{gupta_contributions_1965}.
We split Pasch's axiom A7 into its non-degenerate part A7' and A14, change the version of the parallel postulate A10 and add one axiom (for reasons explained later).
A7' excludes from A7 the degenerate cases where the triangle $ABC$ is flat or when $P$ or $Q$ are respectively not strictly between $A$ and $C$ or $B$ and $C$.
We cannot use A10 as it does not hold in the counter-model found by Szczerba~\cite{szczerba_independence_1970}.
We chose Proclus postulate,\footnote{$\col{A}{B}{C}$ and $\para{AB}{CD}$ denotes that $A$, $B$ and $C$ are collinear and that lines $AB$ and $CD$ are parallel according to the definitions given in SST~\cite{tarski83}.} denoted as A10' in Table~\ref{tarski-variant}, verified to be equivalent to it when assuming A0-A9, using Coq~\cite{Boutry2017}, as it holds in all the counter-models provided by Gupta as well as in the one found by Szczerba, thanks to Theorem~1 in~\cite{szczerba_independence_1970}.
Finally, the formal development found in SST~\cite{tarski83} is essentially classical due to the many case distinctions found in the proofs of its lemmas.
However, the decidability of point equality is sufficient to obtain the arithmetization of geometry in an intuitionistic setting~\cite{BOUTRY2018}.
So we add the decidability of point equality A0 so that we can work in an intuitionistic setting.
The reader not familiar with the difference between classical and intuitionistic logic may refer to~\cite{beeson2015b}.
This system, consisting of A0, A1, A2', A3-A5, A7', A8, A9', A10', A11', A14 and A15, and Tarski's system are equivalent.
Again, we do not detail how to mechanize this fact in Coq.

\section{A model of Tarski's system of geometry}\label{model}

In this section, we present our proof that Cartesian planes over a Pythagorean ordered field form a model of the variant of Tarski's system of geometry that we have introduced in the previous section.
First, we present the structure that we used to define this model.
Then we define the model that we used, that is, the way we instantiated the signature of this system.
Finally, we detail the proofs of some of the more interesting axioms.

\subsection{The \textit{Real Field} Structure}

The structure that was used to define this model was built by Cohen~\cite{cohen_formalized_2012}.
The \textit{real field} structure results of the addition of operators to a discrete\footnote{Discrete fields are fields with a decidable equality.} field: two boolean comparison functions (for strict and non-strict order) and a norm operator.
Elements of this \textit{real field} structure verify the axioms listed in Table~\ref{real-field-axioms}.
Finally, the elements of a \textit{real field} structure are all comparable to zero.
We should remark that this field is not necessarily Pythagorean.
In fact, there is no defined structure in the \textit{Mathematical Components} library~\cite{assia_mahboubi_2022_7118596} for Pythagorean fields.
This can however be added much more easily than before thanks to the recent modification of the \textit{Mathematical Components} library to make use of the \textit{Hierarchy Builder}~\cite{cohen:hal-02478907}.
However, the Pythagorean property is only required for the proof of the segment construction axiom A4.
So we chose to prove that this axiom holds in our model by admitting an extra axiom which was defined in this library: the real closed field axiom.
It states that intermediate value property holds for polynomial with coefficients in the field.
While it is much stronger than Pythagoras' axiom, we only used it to be able to define the square root of a number which is a sum of squares and would therefore have a square root in a Pythagorean field.
Finally, we did not yet prove that $A11$ holds in our model since it would require a much more involved effort.
Indeed, this is similar to verifying that Tarski's system of geometry admits a quantifier elimination procedure.

\begin{table}
\begin{center}
\begin{tabular}{rl}
Subadditivity of the norm operator & $| x + y | \le |x| + |y|$\\
Compatibility of the addition with the strict comparison & $0 < x \land 0 < y \imp 0 < x + y$\\
Definiteness of the norm operator & $| x | = 0 \imp x = 0$\\
Comparability of positive numbers & $ 0 \le x \land 0 \le y \imp (x \le y) || (y \le x)$\\
The norm operator is a morphism for the multiplication & $| x * y | = |x| * |y|$\\
Large comparison in terms of the norm & $(x \le y) = (| y - x | == y - x)$\footnotemark\\
Strict comparison in terms of the large comparison & $(x < y) = (y ~ != x) \&\& (x \le y)$
\end{tabular}
\caption{Axioms of the \textit{real field} structure.}
\label{real-field-axioms}
\end{center}
\end{table}
\footnotetext{$==$ denotes the boolean equality test for the elements of the field.}

\subsection{The Model}

Let us now define our model.
Being based on a single primitive type and two predicates, the signature of Tarski's system of geometry is rather simple.
However, this system has the advantage of having a $n$-dimensional variant.
To obtain this variant, one only needs to change the dimension axioms.
So far, we have restricted ourselves to the planar version of this system.
With a view to extend the \textit{GeoCoq} library to its $n$-dimensional variant, we wanted to define a model in which we could prove all but the dimension axioms in an arbitrary dimension to be able to construct a model of the $n$-dimensional variant by only proving the new dimension axioms.
Hence we chose to define \texttt{Tpoint} as a vector of dimension $n+1$ with coefficient in the \textit{real field} structure $\mathbb{F}$ (we used the \textit{real field} structure for all the development with the exception of the proof of the segment construction axiom) for a fixed integer $n$, that is \texttt{'rV[R]\_(n.+1)}.
We adopted Gupta's definition~\cite{gupta_contributions_1965} for the congruence \texttt{cong}, namely that $\congT{AB}{CD}$ if the squares of the Euclidean norms of $B-A$ and $D-C$ are equal.
Actually Gupta also proved that any model of the $n$-dimensional variant of Tarski's system of geometry is isomorphic to his model.
He defined that $\bT{A}{B}{C}$ holds if and only if there exists a $k \in \mathbb{F}$ such that $0 \le k \le 1$ and $B-A = k (C-A)$.
In fact, if such a $k$ exists, it can be computed.
By letting $A = \left( a_i \right)_{1 \le i \le n+1}$, $B = \left( b_i \right)_{1 \le i \le n+1}$ and $C = \left( c_i \right)_{1 \le i \le n+1}$, if $A \neq C$ then there exists a $i \in \mathbb{N}$ such that $1 \le i \le n+1$ and $a_i \neq c_i$ and in this case we set $k$ to $\frac{b_i - a_i}{c_i - a_i}$ and if $A = C$ we set $k$ to zero.
Therefore we defined a function \texttt{ratio} that computes the possible value for $k$, thus allowing us to define the betweenness by the boolean equality test.
This was actually important as it permitted to directly manipulate the definition for betweenness by rewriting since we defined it as a boolean test.
Finally, as it was often necessary to distinguish whether or not $\bT{A}{B}{C}$ holds due to a degeneracy, we split the definition \texttt{bet} of the betweenness into two predicates: the first one, \texttt{betS}, capturing the general case of $k$ being strictly between $0$ and $1$ and the second one, \texttt{betE}, capturing the three possible degenerate cases, namely either $A = B$, $B = C$ or $A = B$ and $B = C$.

Formally, we consider the following model:

\begin{CoqVerbatim}
Variable R : realFieldType.
Variable n : nat.

Implicit Types (a b c d : 'rV[R]_(n.+1)).

Definition cong a b c d := (b - a) *m (b - a)^T == (d - c) *m (d - c)^T.

Definition betE a b c := [ || [ && a == b & b == c ], a == b | b == c ].

Definition ratio v1 v2 :=
  if [pick k : 'I_(n.+1) | v2 0 k != 0] is Some k
  then v1 0 k / v2 0 k else 0.

Definition betR a b c := ratio (b - a) (c - a).

Definition betS a b c (r := betR a b c) :=
  [ && b - a == r *: (c - a), 0 < r & r < 1].

Definition bet a b c := betE a b c || betS a b c.
\end{CoqVerbatim}

\subsection{Proof that the Axioms hold in the Model}

Now that we have defined the model, we focus on the proof that the axioms of the system from Sec.~\ref{independent-tarski} hold in this model.
However, we omit the details of the proofs for axioms A1, A2', A3 and A14 since they are rather straightforward.
For the same reason, we do not cover the decidability of point equality A0.

Let us start by focusing on axioms A7' and A15 as the proofs that they hold in our model are quite similar.
In the case of axiom A15 we know that $\bT{A}{B}{D}$ and $\bT{B}{C}{D}$ so let $k_1 \in \mathbb{F}$ be such that $0 < k_1 < 1$ and $B - A = k_1 (D - A)$ (the degenerate case of this axiom is trivial so we only consider the general case) and $k_2 \in \mathbb{F}$ be such that $0 < k_2 < 1$ and $C - B = k_2 (D - B)$.
In order to prove that $\bT{A}{B}{C}$ we need to find a $k \in \mathbb{F}$ such that $0 < k < 1$ and $B - A = k (C - A)$.
By calculation we find that $k = \frac{k_1}{k_1 + k_2 - k_1 k_2}$ and we can verify that $0 < k < 1$.
In a similar way, for axiom A7', we know that $\bT{A}{P}{C}$ and $\bT{B}{Q}{C}$ so let $k_1 \in \mathbb{F}$ be such that $0 < k_1 < 1$ and $P - A = k_1 (C - A)$ (the hypotheses imply that $0 < k_1 < 1$ because $A \neq P$ and $P \neq C$) and $k_2 \in \mathbb{F}$ be such that $0 < k_2 < 1$ and $Q - B = k_2 (C - B)$.
In order to prove that there exists a point $X$ such that $\bT{P}{X}{B}$ and $\bT{Q}{X}{A}$ we need to find a $k_3 \in \mathbb{F}$ and a $k_4 \in \mathbb{F}$ such that $0 < k_3 < 1$, $0 < k_4 < 1$ and $k_3 (B - P) + P = k_4 (A - Q) + Q$.
By calculation we find that $k_3 = \frac{k_1 (1 - k_2)}{k_1 + k_2 - k_1 k_2}$ and $k_4 = \frac{k_2 (1 - k_1)}{k_1 + k_2 - k_1 k_2}$ and we can verify that $0 < k_3 < 1$ and $0 < k_4 < 1$.
In both of these proof, the ratios are almost identical to the point that it suffices to prove the following lemma:

\begin{CoqVerbatim}
Lemma ratio_bet a b c k1 k2 k3 :
  0 < k1 -> 0 < k2 -> k1 < 1 -> 0 < k3 -> k3 < k1+k2-k1*k2 ->
  b - a == ((k1+k2-k1*k2)/k3)^-1 *: (c - a) -> bet a b c.
\end{CoqVerbatim}

It allows to prove quite easily both of these axioms.
For axiom A4, we proceeded in a analogous way: it suffices to set the point $E$ that can be constructed using this axiom to $\frac{\| D - C \|}{\| B - A \|} (B - A) + A$ and to verify this point satisfies the desired properties by calculation.

We now turn to axiom A5.
We followed Makarios' approach for the proof that this axiom holds in our model~\cite{makarios}.
In his proof he used the cosine rule: in a triangle whose vertices are the vectors $A$, $B$ and $C$ we have
$$\| C - B \|^2 = \| C - A \|^2 + \| B - A \|^2 - 2 (B - A) \cdot (C - A).$$
As noted by Makarios, using the cosine rule allows to avoid defining angles and properties about them.
Applying the cosine rule for the triangles $BCD$ and $B'C'D'$ allows to prove that $\| D - C \|^2 = \| D' - C' \|^2$ by showing that
$$(C - B) \cdot (D - B) = (C' - B') \cdot (D' - B')$$
which can be justified, by applying the cosine rule again, this time in the triangles $ABD$ and $A'B'D'$, if
$$\| D - A \| - \| D - B \| - \| A - B \| = \| D' - A' \| - \| D' - B' \| - \| A' - B' \|$$
which we know from the hypotheses and if the ratios corresponding to the betweenness $\bT{A}{B}{C}$ and $\bT{A'}{B'}{C'}$ are equal which can be obtained by calculation.

Next, let us consider axiom A10.\footnote{
Similarly to A7, when we were proving Euclid's axiom, we realized that the same kind of distinctions was also needed.
The degenerate cases are implied by the other betweenness axioms so it suffices to show that A10 holds when the angle $\angle BAC$ is non-flat and when $D$ is different from $T$.}
From the hypotheses we have two ratios $k_1 \in \mathbb{F}$ and $k_2 \in \mathbb{F}$ such that $0 < k_1 < 1$, $0 < k_2 < 1$, $D - A = k_1 (T - A)$ and $D - B = k_2 (C - B)$.
Using these ratios, it suffices to define $X$ such that $B - A = k_1 (X - A)$ and $Y$ such that $C - A = k_1 (Y - A)$.
So we know by construction that $\bT{A}{B}{X}$ and $\bT{A}{C}{Y}$ and we easily get that $T - X = k_2 (Y - X)$ by calculation, thus proving that $\bT{X}{T}{Y}$.
Since A10 and A10' are equivalent when A0, A1, A2', A3-A5, A7', A8, A9', A11', A14 and A15 hold, this allows to prove that A10' holds in our model.

Finally the remaining two axioms are treated in a slightly different setting since they are the dimension axioms.
Formally we fix the value of $n$ to $1$.
In order to simplify the many rewriting steps needed for these proofs we started by establishing the following two lemmas:

\begin{CoqVerbatim}
Definition sqr_L2_norm_2D a b :=
  (b 0 0 - a 0 0) ^+ 2 + (b 0 1 - a 0 1) ^+ 2.

Lemma congP a b c d :
  reflect (sqr_L2_norm_2D a b = sqr_L2_norm_2D c d) (cong a b c d).

Lemma betSP' a b c (r := betR a b c) :
  reflect ([ /\ b 0 0 - a 0 0 = r * (c 0 0 - a 0 0),
               b 0 1 - a 0 1 = r * (c 0 1 - a 0 1), 0 < r & r < 1])
          (betS a b c).
\end{CoqVerbatim}

The reader familiar with \textsc{SSReflect} will have recognized the \texttt{reflect} predicate, described in~\cite{cohen_formalized_2012} for example.
In practice, these lemmas allowed to spare many steps that would have been repeated in almost every proof concerning the dimension axioms.
It was much more straightforward to prove that axiom A8 holds in our model than for axiom A9'.
In fact, it is enough to find three non-collinear points.
We simply took the points $(0, 0)$, $(0, 1)$ and $(1, 0)$:

\begin{CoqVerbatim}
Definition row2 {R : ringType} (a b : R) : 'rV[R]_2 :=
  \row_p [eta \0 with 0 |-> a, 1 |-> b] p.

Definition a : 'rV[R]_(2) := row2 0 0.
Definition b : 'rV[R]_(2) := row2 0 1.
Definition c : 'rV[R]_(2) := row2 1 0.
\end{CoqVerbatim}

It was then an easy matter to verify that axiom A8 holds in our model.
For axiom A9, the idea of the proof that we formalized was to first show that, by letting $M$ be the midpoint of $P$ and $Q$, the equation $(x_P - x_M) (x_M - x_X) + (y_P - y_M) (y_M - y_X) = 0$, capturing the property that the points $P$, $M$, and $X$ form a right angle with the right angle at vertex $M$, was verified when $X$ would be equal to $A$, $B$ or $C$:

\begin{CoqVerbatim}
Lemma cong_perp (a p q : 'rV[R]_(2)) (m := (1 / (1 + 1)) *: (p + q)) :
  cong a p a q ->
  (p 0 0 - m 0 0) * (m 0 0 - a 0 0) +
  (p 0 1 - m 0 1) * (m 0 1 - a 0 1) = 0.
\end{CoqVerbatim}

Next, we demonstrated that for three points $A$, $B$ and $C$ verifying $(x_A - x_B) (y_B - y_C) - (y_A - y_B) (x_B - x_C) = 0$ are collinear in the sense that $\bT{A}{B}{C} \lor \bT{B}{C}{A} \lor \bT{C}{A}{B}$:

\begin{CoqVerbatim}
Lemma col_2D a b c :
  (a 0 0 - b 0 0) * (b 0 1 - c 0 1) ==
  (a 0 1 - b 0 1) * (b 0 0 - c 0 0) ->
  (bet a b c \/ bet b c a \/ bet c a b).
\end{CoqVerbatim}

Using the equations implied by \verb|cong_perp| we could derive that
$$(x_P - x_M) (y_M - y_P) \left( (x_A - x_B) (y_B - y_C) - (y_A - y_B) (x_B - x_C) \right) = 0.$$
We were then left with three cases: either the abscissas of $P$ and $M$ are equal in which case the ordinate of $A$, $B$ and $C$ were equal thus sufficing to complete the proof, or the ordinates of $P$ and $M$ are equal in which case the abscissas of $A$, $B$ and $C$ were equal thus completing the proof, or $(x_A - x_B) (y_B - y_C) - (y_A - y_B) (x_B - x_C) = 0$ corresponding to the lemma that we had proved and again allowing to conclude.

Putting everything together, we could prove that Cartesian planes over a Pythagorean ordered field form a model of the variant of Tarski's system of geometry, thus proving the satisfiability of the theory.\footnote{\texttt{Tarski\_euclidean} is the type class that captures the theory consisting of axioms A0-A10.}

\begin{CoqVerbatim}
Global Instance Rcf_to_T2D : Tarski_2D Rcf_to_T_PED.

Global Instance Rcf_to_T_euclidean : Tarski_euclidean Rcf_to_T_PED.
\end{CoqVerbatim}

\section{An Example of Independence Proof}\label{klein}

To illustrate how we obtain formal proofs of independence we present an example.
We start by defining the counter-model we will use to prove the independence of axiom A10'.
We then provide the sketch of the formal proof.

\subsection{Klein's Model}

To prove Euclid’s Parallel Postulate independent from the other axiom we work in Klein's model as defined in SST~\cite{tarski83}:

\begin{CoqVerbatim}
Variable R : realFieldType.
Variable n : nat.

Definition Vector := 'rV[R]_(n.+1).
Definition Point : Type := {p : Vector | (p *m p^T) 0 0 < 1}.
Notation "#" := proj1_sig.

Implicit Types (a b c d : Point).
Implicit Types (v w x y : Vector).

Definition bet' a b c := bet (#a) (#b) (#c).

Definition omd_v v w := (1 - (v *m (w)^T) 0 0).
Definition cong_v v w x y :=
  (omd_v v w)^+2/(omd_v v v * omd_v w w) ==
  (omd_v x y)^+2/(omd_v x x * omd_v y y).
Definition cong' a b c d := cong_v (#a) (#b) (#c) (#d).
\end{CoqVerbatim}

Here, \texttt{Point} is the type of \texttt{Vector}, vectors of dimension $n+1$ with coefficient in the \textit{real field} structure, lying inside the unit disk and \texttt{\#} the projection allowing to recover the coordinate part of this dependent type.
In Klein's model, $b$ is said to be between $a$ and $c$ iff their coordinate parts can be said to be \texttt{bet} in the model from Sec.~\ref{model} and line-segments $\overline{ab}$ and $\overline{cd}$ are said to be congruent iff 
\begin{small}
$$\frac{\left( 1 - \# a \cdot \# b \right)^2}{\left( 1 - \# a \cdot \# a \right) \left( 1 - \# b \cdot \# b \right)} = \frac{\left( 1 - \# c \cdot \# d \right)^2}{\left( 1 - \# c \cdot \# c \right) \left( 1 - \# d \cdot \# d \right)}$$
\end{small}
where $\cdot$ denotes the dot product of two vectors.

\subsection{Independence of Euclid’s Parallel Postulate via Klein's Model}

Here we only detail the proof that A10' does not hold in this model.
Mechanizing the following proof sketch allows to derive.

\begin{CoqVerbatim}
Lemma euclid : ~ euclidP (@Point R 1) (@bet' R 1).
\end{CoqVerbatim}

To make sure that we did not introduce any change in the axioms between the various models we relied on predicates such as \texttt{euclidP}, which depend on possibly the type for points and the predicate(s) for betweenness and/or congruence.

\begin{theorem}\label{klein-parallel}
Axiom A10' does not hold in Klein's model.
\end{theorem}

\begin{proof}
Since Klein's model forms a model of neutral geometry,\footnote{Neutral geometry is defined by the set of axioms of Euclidean geometry from which the parallel postulate has been removed.} it suffices to prove that any version of the parallel postulate, proven equivalent to A10' in Coq when assuming A0-A9, does not hold.
We choose to work with A10.
Picking $a$, $b$, $c$, $d$ and $t$ to be of coordinates $(0, 0)$, $(0, \frac{1}{2})$, $(\frac{1}{2}, 0)$, $(\frac{1}{4}, \frac{1}{4})$ and $(\frac{1}{2}, \frac{1}{2})$, some computations allow to verify that $\bT{\# a}{\# d}{\# t}$, $\bT{\# b}{\# d}{\# c}$ $b \neq d$, $d \neq c$ and $\lnot \col{\# a}{\# b}{\# c}$.
So, to prove that this version does not hold, it is enough to show that for any $x$ and $y$ such that $x$ lies inside the unit disk, $\bT{\# a}{\# b}{\# x}$, $\bT{\# a}{\# c}{\# y}$ and $\bT{\# x}{\# t}{\# y}$, it holds that $y$ is not a \texttt{Point}, meaning that it lies outside the unit disk.
Let us first eliminate the case where $b = x$ as it would lead to a contradiction.
Here, we use the algebraic characterization of collinearity\footnote{Here we use the converse of \texttt{col\_2D}.} to obtain that, if $b = x$, the ordinate of x would need to be equal to both $0$ and $\frac{1}{2}$ which is impossible.
Now let us pose $b'$ to be the vector $x + a - b$.
It is an easy matter to check that $\bT{\# a}{\# b'}{\# x}$ so let us pose $k_1$ to be the ratio associated to this betweenness.
We can verify that $k_1 \le \frac{1}{2}$ since $x$ is supposed to belong to the unit disk.
We can then take $d'$ at ratio $k_1$ from $a$ to $t$.
Applying what was proven to show that A10' holds in Cartesian planes over a Pythagorean ordered field, we can show that $y'$ at ratio $\frac{1}{k_1}$ from $a$ to $c$ is such that $\bT{\# a}{\# c}{\# y'}$ and $\bT{\# x}{\# t}{\# y'}$.
If we can prove that $y = y'$ we will be done as $y'$ lies outside of the unit disk because $k_1 \le \frac{1}{2}$ so $2 \le \frac{1}{k_1}$.
Finally, to prove that $y = y'$ we can reason by uniqueness of the intersection of lines which is valid in neutral geometry.
\end{proof}

\section{Conclusion}

We defined ten out of the eleven counter-models present in Gupta's thesis~\cite{gupta_contributions_1965}, thus obtaining the Coq formal proof of the independence of ten out of the thirteen axioms of the system presented in Sec.~\ref{independent-tarski}.
This seems to indicate that Pasch's axiom could indeed be split into two meaningfully different parts as done in this paper while still having an independent system.
However, we will only be sure of this once we will have formalized the missing three counter-models.
These can be found in Gupta's thesis~\cite{gupta_contributions_1965}, Szczerba's paper~\cite{szczerba_independence_1970}, and Beeson's section \textit{The recursive model} in~\cite{beeson2015a}.

Five of the formalized models are finite and the other five are modifications of the model presented in Sec.~\ref{model}.
We highlight that, for the latter five, A11' is not verified for the same reason as for the model from Sec.~\ref{model}..
All these models are available in the \textit{GeoCoq} library\footnote{\url{http://geocoq.github.io/GeoCoq/}} and represent about 4k lines of formal proof.

We are currently extending this work by proving the independence of a more constructive version\footnote{We replace point equality decidability by point equality \enquote{stability}, namely $\forall X Y, \lnot \lnot X = Y \imp X = Y$, which allows to prove equality of points by contradiction but does not allow case distinctions. We do not go as far as in~\cite{beeson_brouwer_2017} where not even \enquote{stability} is assumed. We also apply the same modifications made to obtain what is called \enquote{continuous Tarski geometry} in~\cite{beeson2015b}.} of the axioms which would also allow to capture $n$-dimensional geometry.
For this extension we could not rely on A9$^{(n)}$ from~\cite{givant99}.
Indeed, we found that it can only be assumed as an upper $n$-dimensional axiom when $n =2$ or $3$.
A9$^{(n)}$ is stated as follows.
$$ \bigwedge\limits_{1 \le i \le j \le n} P_i \neq P_j \land \bigwedge\limits_{i = 2}^n \congT{A P_1}{A P_i} \land \bigwedge\limits_{i = 2}^n \congT{B P_1}{B P_i} \land \bigwedge\limits_{i = 2}^n \congT{C P_1}{C P_i} \imp \col{A}{B}{C}$$
By taking $P_i = (\cos \frac{2i \pi}{n}, \sin \frac{2i \pi}{n}, 0, \ldots, 0)$ for $1 \le i \le n$ then $(0, 0, x_3, x_4, \ldots, x_n)$ satisfies the premises for any $x_3, x_4, \ldots, x_n$ in the standard n-dimensional model while triplets of points of this form are not necessarily collinear.
The various modifications did not allow to reuse some of the counter-models already mechanized, so new ones are necessary.

We are convinced that using a proof-assistant is crucial when proving the independence of a system, where small changes in a statement are critical.
Actually, there was a typo in Gupta's counter-model for A2 and we just exhibited a problem with axiom A9$^{(n)}$ from~\cite{givant99}.
The \textit{GeoCoq} library also proved very useful as it allowed us to combine the algebraic and geometric\footnote{Such as the use of the uniqueness of the intersection in our proof of Theorem~\ref{klein-parallel}.} reasoning.\\

\textbf{Acknowledgments:} We would like to thank Marius Hinge for his contribution to the early stage of this work.

\bibliographystyle{eptcs}
\bibliography{biblio}

\newcommand{\noopsort}[1]{}\newcommand{\singleletter}[1]{#1}
\begin{thebibliography}{10}
\providecommand{\bibitemdeclare}[2]{}
\providecommand{\surnamestart}{}
\providecommand{\surnameend}{}
\providecommand{\urlprefix}{Available at }
\providecommand{\url}[1]{\texttt{#1}}
\providecommand{\href}[2]{\texttt{#2}}
\providecommand{\urlalt}[2]{\href{#1}{#2}}
\providecommand{\doi}[1]{doi:\urlalt{http://dx.doi.org/#1}{#1}}
\providecommand{\bibinfo}[2]{#2}


\bibitem{beeson2015b}
Beeson, M.: {A Constructive Version of Tarski's Geometry}. Annals of Pure and
  Applied Logic  \textbf{166}(11),  1199--1273 (2015).
  \doi{10.1016/j.apal.2015.07.006}

\bibitem{beeson2015a}
Beeson, M.: {Constructive Geometry and the Parallel Postulate}. Bulletin of
  Symbolic Logic  \textbf{22}(1),  1--104 (2016). \doi{10.1017/bsl.2015.41}

\bibitem{beeson_brouwer_2017}
Beeson, M.: {Brouwer and Euclid}. Indagationes Mathematicae  \textbf{29}(1),
  483--533 (2018). \doi{10.1016/j.indag.2017.06.002}

\bibitem{beeson_herbrands_2015}
Beeson, M., Boutry, P., Narboux, J.: {Herbrand’s theorem and non-Euclidean
  geometry}. The Bulletin of Symbolic Logic  \textbf{21}(2),  111--122 (2015).
  \doi{10.1017/bsl.2015.6}

\bibitem{bolyai1832appendix}
Bolyai, J.: {Appendix, Scientiam Spatii absolute veram exhibens: a veritate aut
  falsitate Axiomatis XI. Euclidei (a priori haud unquam decidenda)
  independentem; adjecta ad casum falsitatis, quadratura circuli geometrica.
  Auctore Johanne Bolyai de eadem, Geometrarum in Exercitu Caesareo Regio
  Austriaco Castrensium Capitaneo}. Coll. Ref. (1832)

\bibitem{BOUTRY2018}
Boutry, P., Braun, G., Narboux, J.: {Formalization of the Arithmetization of
  Euclidean Plane Geometry and Applications}. Journal of Symbolic Computation
  (2018). \doi{10.1016/j.jsc.2018.04.007}

\bibitem{Boutry2017}
Boutry, P., Gries, C., Narboux, J., Schreck, P.: {Parallel Postulates and
  Continuity Axioms: A Mechanized Study in Intuitionistic Logic Using Coq}.
  Journal of Automated Reasoning  (2017),
  \doi{10.1007/s10817-017-9422-8}

\bibitem{braun:hal-01332044}
Braun, G., Boutry, P., Narboux, J.: {From Hilbert to Tarski}. In: Narboux, J.,
  Schreck, P., Streinu, I. (eds.) {Proceedings of the Eleventh International
  Workshop on Automated Deduction in Geometry}. pp. 78--96 (2016),
  \url{https://hal.inria.fr/hal-01332044}

\bibitem{cohen_formalized_2012}
Cohen, C.: {Formalized algebraic numbers: construction and first-order theory}.
  Theses, Ecole Polytechnique X (Nov 2012),
  \url{https://pastel.archives-ouvertes.fr/pastel-00780446}

\bibitem{cohen:hal-02478907}
Cohen, C., Sakaguchi, K., Tassi, E.: {Hierarchy Builder: Algebraic hierarchies
  Made Easy in Coq with Elpi (System Description)}. In: 5th International
  Conference on Formal Structures for Computation and Deduction (FSCD 2020).
  vol.~167, pp. 34:1--34:21. Paris, France (Jun 2020).
  \doi{10.4230/LIPIcs.FSCD.2020.34}

\bibitem{gupta_contributions_1965}
Gupta, H.N.: {Contributions to the Axiomatic Foundations of Geometry}. Ph.D.
  thesis, University of California, Berkeley (1965)

\bibitem{hilbert_les_1971}
Hilbert, D.: {Les fondements de la g{\'e}om{\'e}trie}. {Dunod}, {Jacques Gabay}
  edn. (1971), edition critique avec introduction et compl{\'e}ments
  pr{\'e}par{\'e}e par Paul Rossier

\bibitem{klugel63}
Klugel, G.S.: {Conatuum praecipuorum theoriam parallelarum demonstrandi
  recensio}. Ph.D. thesis, Schultz, Göttingen (1763), German translation
  available: \url{http://www2.math.uni-wuppertal.de/~volkert/versuch.html}

\bibitem{lobachevsky1840geometrische}
Lobatschewsky, N.: {Geometrische Untersuchungen zur Theorie der
  Parallellinien}, pp. 159--223. Springer Vienna, Vienna (1985).
  \doi{10.1007/978-3-7091-9511-6\_4}

\bibitem{assia_mahboubi_2022_7118596}
Mahboubi, A., Tassi, E.: {Mathematical Components}. Zenodo (Sep 2022),
  \doi{10.5281/zenodo.7118596}

\bibitem{makarios}
Makarios, T.J.M.: {A Mechanical Verification of the Independence of Tarski's
  Euclidean Axiom}. Master's thesis, Victoria University of Wellington (2012)

\bibitem{makarios_further_2013}
Makarios, T.J.M.: A further simplification of {Tarski}'s axioms of geometry.
  arXiv: Logic  (2013). \doi{10.1285/i15900932v33n2p123}.

\bibitem{narboux:hal-01779452}
Narboux, J., Janicic, P., Fleuriot, J.: {Computer-assisted Theorem Proving in
  Synthetic Geometry}. In: Sitharam, M., John, A.S., Sidman, J. (eds.)
  {Handbook of Geometric Constraint Systems Principles}. Discrete Mathematics
  and Its Applications, {Chapman and Hall/CRC } (2018).
  \doi{10.1201/9781315121116-2}, \url{https://inria.hal.science/hal-01779452}

\bibitem{tarski83}
Schwabh{\"a}user, W., Szmielew, W., Tarski, A.: {Metamathematische Methoden in
  der Geometrie}. Springer-Verlag (1983)

\bibitem{simic-poincare}
Simi{\'c}, D., Mari{\'c}, F., Boutry, P.: {Formalization of the Poincar{\'e}
  Disc Model of Hyperbolic Geometry}. {Journal of Automated Reasoning}
  \textbf{65}(1),  31--73 (Apr 2020). \doi{10.1007/s10817-020-09551-2}

\bibitem{szczerba_independence_1970}
Szczerba, L.W.: Independence of {Pasch}'s axiom. Bull. Acad. Polon. Sci. Sér.
  Sci. Math. Astronom. Phys.  \textbf{18},  491--498 (1970)

\bibitem{givant99}
Tarski, A., Givant, S.: {Tarski's System of Geometry}. The Bulletin of Symbolic
  Logic  \textbf{5}(2),  175--214 (1999). \doi{10.2307/421089}

\bibitem{the_coq_development_team_2018_1174360}
Team, T.C.D.: {The Coq Proof Assistant, version 8.15} (2022),
  \doi{10.5281/zenodo.5846982}

\end{thebibliography}

\end{document}